\documentclass[a4paper,12pt]{article}

\usepackage{amsthm}
\usepackage{amsmath}
\usepackage{amssymb}
\usepackage{graphicx}
\usepackage[round]{natbib}
\usepackage{bbm}
\usepackage{url}

\theoremstyle{plain}  %default 
\newtheorem{thm}{Theorem}[section] 
 
\newtheorem{prop}[thm]{Proposition} 
\newtheorem{cor}[thm]{Corollary} 

\theoremstyle{definition} 
\newtheorem{defn}{Definition}[section]

\theoremstyle{remark} 
\newtheorem*{rem}{Remark}

\newcommand{\prob}{\mathbf{P}}

\newcommand{\E}{\mathbb{E}}

\newcommand{\essinf}{\operatorname{ess\, inf}}
\newcommand{\VaR}{\operatorname{VaR}}
\newcommand{\ES}{\operatorname{ES}}
\newcommand{\argmin}{\operatorname{arg\, min}}

\begin{document}

\title{Coherence and elicitability}
\author{Johanna F.~Ziegel\thanks{The author would like to thank Paul Embrechts, Tilmann Gneiting and Fabio Bellini for discussions. \newline
Address correspondence to Johanna F.~Ziegel, University of Bern, Department of Mathematics and Statistics,
Institute of Mathematical Statistics and Actuarial Science, Sidlerstrasse 5, 3012 Bern, Switzerland, e-mail:\texttt{johanna.ziegel@stat.unibe.ch}}}
\date{}
\maketitle

\begin{abstract}
The risk of a financial position is usually summarized by a risk measure. As this risk measure has to be estimated from historical data, it is important to be able to verify and compare competing estimation procedures. In statistical decision theory, risk measures for which such verification and comparison is possible, are called elicitable. It is known that quantile based risk measures such as value at risk are elicitable. In this paper, Gneiting's (2011) result of the non-elicitability of expected shortfall is extended to all law-invariant spectral risk measures unless they reduce to minus the expected value. Hence, it is unclear how to perform forecast verification or comparison. However, the class of elicitable law-invariant coherent risk measures does not reduce to minus the expected value. We show that it consists of certain expectiles.
\end{abstract}

\section{Introduction}\label{sec:intro}

Value at Risk (VaR) is the most common risk measure used in banking and finance. The \emph{VaR at level} $\alpha \in (0,1)$ is given by 
\[
\VaR_{\alpha}(Y) = -\inf\{x \in \mathbb{R}\;|\; F_Y(x) \ge \alpha\},
\]
where the financial position $Y$ is a real-valued random variable, and $F_Y$ is its cumulative distribution function.  In this paper, a positive value of $Y$ denotes a profit. The sign convention we have chosen for $\VaR$ implies that extreme losses correspond to levels $\alpha$ close to zero, and for $Y \le 0$, the risk $\VaR_{\alpha}(Y)$ will be non-negative. 
Since the influential paper of \citet{ArtznerDelbaenETAL1999} introduced coherent risk measures, VaR has frequently been criticized as a risk measure because it fails to be subadditive, and hence it is not coherent; see for example \citet{Acerbi2002}. Other authors have pointed out the lack of VaR at level $\alpha$ to account for the size of losses beyond the level $\alpha$  \citep{DanielssoEmbrechtsETAL2001}. Median shortfall at level $\alpha$, or equivalently, VaR at level $(\alpha+1)/2$ \citep{KouPengETAL2013}, does account for the size of losses beyond level $\alpha$. 

The \citet{Basel2012} has been investigating the points in favor of and against a change of the regulatory risk measure from VaR to the coherent risk measure \emph{expected shortfall (ES)}, also known as \emph{average} or \emph{conditional value at risk}, which is defined by
\[
\ES_{\alpha}(Y)=\frac{1}{\alpha}\int_0^{\alpha} \VaR_{\tau}(Y)d\tau.
\]
From the perspective of coherent risk measures, ES is a better alternative to VaR. It remedies both problems mentioned above: It is a coherent risk measure, and it is sensitive to the sizes of the potential losses beyond the threshold $\alpha$. Other popular coherent risk measures are the so-called \emph{spectral risk measures}, which generalize ES \citep{Acerbi2002}.

However, despite their theoretical appeal, there are also major drawbacks to using spectral risk measures in risk management, which should not be neglected. \citet{ContDeguestETAL2010} show that there is a fundamental theoretical conflict between subadditivity and robustness of risk measurement procedures for spectral risk measures; see also the related discussion in \citet[Section 5]{KouPengETAL2013}. \citet{ContDeguestETAL2010}  state
\begin{quote}
We hope to have convinced the reader that there is more to risk measurement than the choice of a `risk measure': statistical robustness, and not only `coherence', should be a concern for regulators and end-users when choosing or designing risk measurement procedures. The design of robust risk estimation procedures requires the explicit inclusion of the statistical estimation step in the analysis of the risk measurement procedure.
\end{quote}
The next steps beyond estimation are backtesting and forecast verification.  Backtesting refers to validating a given estimation procedure for a risk measure on historical data. In this paper, following the ideas of \citet{Gneiting2011}, we consider risk measures from a forecasting perspective. With our knowledge of today, we are trying to give the best possible point estimate of the risk measure for tomorrow, or ten days ahead, or for any other time point in the future. There are numerous choices concerning models, methods and parameters that have to be made to come up with predictions. Hence, for a number of competing forecast or estimation procedures we would like to decide which one performs best.  If we restrict our attention to law-invariant coherent risk measures as introduced by \citet{Kusuoka2001} we can view them as functionals on some set $\mathcal{P}$ of probability distributions on $\mathbb{R}$. From the viewpoint of statistical decision theory not all functionals allow for meaningful point forecasts; see \citet{Gneiting2011}. Functionals for which meaningful point forecasts and forecast performance comparisons are possible are called \emph{elicitable}; see Section \ref{sec:elicit} for details. One important example of elicitable functionals are quantiles, hence VaR is elicitable. 

\citet{Gneiting2011} has shown that ES is not elicitable, which may be a partial explanation for the difficulties with robust estimation and backtesting. This raises the natural question whether there is a  different  option. Is there any (interesting) law-invariant coherent risk measure that is also an elicitable functional? We show that the only law-invariant spectral risk measure that is also elicitable is minus the expected value:
\[
\rho(Y) = -\E[Y];
\]
see Corollary \ref{cor:2}. However, there are law-invariant coherent risk measures that are elicitable. They are \emph{expectiles} which were first introduced by \citet{NeweyPowell1987}. The elicitability of expectiles is a simple corollary of their definition. They have been considered as a risk measure by \citet{KuanYehETAL2009}. Proposition \ref{prop:4.4} shows that they are coherent risk measures. Very recently, and independently of our work, a proof of this result also appears in \citet{BelliniKlarETAL2013}. Proposition \ref{prop:4.4} also identifies the minimal generating set of the Kusuoka representation as defined in \citet[Definition 2.3]{PichlerShapiro2012}. Expectiles are the only elicitable law-invariant coherent risk measures; see Section \ref{sec:expectiles}.

In the literature, there are procedures for evaluating ES forecasts and that allow for tests; see for example \citet{McNeilFrey2000,Christoff2003}. However, these methods do not allow for a direct comparison and ranking of the predictive performance of competing forecasting methods \citep{Gneiting2011}.

The non-elicitability of spectral risk measures, and in particular of ES, is the reason that there is no analogue to the quantile regression method \citep{Koenker2005} for these functionals, and no M-estimators can be constructed. \citet{ChunShapiroETAL2012} construct a mixed quantile estimator for ES as an approximation to an M-estimator. The recent contribution of \citet{RockafellRoysetETAL2013} takes this approach further and proposes a framework for `generalized' regression that is suitable for ES. However, the problem with forecast comparison remains. 

The paper is organized as follows. In Section \ref{sec:elicit} we introduce the notion of elicitability and describe its importance in point forecasting. A brief introduction to law-invariant coherent risk measures is given in Section \ref{sec:riskmeasures}. Section \ref{sec:coherelicit} contains the main results of the paper, showing in particular that law-invariant spectral risk measures are not elicitable, and elaborating the prominent role of expectiles as the only elicitable law-invariant coherent risk measures. We conclude the paper with a discussion; see Section \ref{sec:disc}.

\section{Elicitability}\label{sec:elicit}

Let $\mathcal{P}$ be a class of probability measures on $\mathbb{R}$ with the Borel sigma algebra. We consider a functional 
\[
\nu:\mathcal{P} \to 2^{\mathbb{R}}, \quad P \mapsto \nu(P) \subset \mathbb{R},
\]
where $2^{\mathbb{R}}$ denotes the power set of $\mathbb{R}$. Often, but not always, $\nu(P)$ is single valued, for example if we consider the expectation functional $\E$ on the class of all probability measures with finite mean. However, quantile functionals may be set-valued. In the case of single valued functionals we will confound the one-point set $\nu(P)$ with its unique element. 

In this paper we are interested in the statistical properties of functionals that are law-invariant coherent risk measures; see Section \ref{sec:riskmeasures}.
The following  Definitions \ref{def:consist} and \ref{def:elicit} are central in the context of point forecasting; see \citet[Section 2]{Gneiting2011} for a discussion of their historical background. Let $Y$ be a real-valued random variable, which models the future observation of interest. 

\begin{defn}\label{def:consist}
A scoring function $s:\mathbb{R}\times\mathbb{R} \to [0,\infty)$ is \emph{consistent} for the functional $\nu$ relative to the class $\mathcal{P}$, if 
\begin{equation}\label{eq:consistency}
\E_P s(t,Y) \le \E_P s(x,Y) 
\end{equation}
for all $P \in \mathcal{P}$, all $t \in \nu(P)$, and all $x \in \mathbb{R}$. Here, $Y$ has distribution $P$. It is \emph{strictly consistent} if it is consistent and equality in \eqref{eq:consistency} implies that $x \in \nu(P)$.
\end{defn}

Given a consistent scoring function $s$ for a functional $\nu$, an \emph{optimal forecast} $\hat{x}$ for $\nu(P)$ is given by
\[
\hat{x} = \argmin_{x} \E_P s(x,Y).
\]
Competing forecast procedures for $\nu$ can be compared using the scoring function $s$. Suppose that in $n$ forecast cases we have point forecasts $x^{(k)}_1, \dots, x^{(k)}_n$, $k=1,\dots,K$, and realizing observations $y_1,\dots,y_n$. The index $k$ numbers the $K$ competing forecast procedures. We can rank the procedures by their average scores
\[
\bar{s}^{(k)} = \frac{1}{n}\sum_{i=1}^n s\big(x^{(k)}_i,y_i\big).
\]
The consistency of the scoring rule for the functional $\nu$ ensures that accurate forecasts of $\nu(P)$ are rewarded. On the contrary, evaluating point forecasts with respect to `some' scoring function, which is not consistent for $\nu$, may lead to grossly misguided conclusions about the quality of the forecasts. A drastic example is provided in the simulation study of \citet[Section 1.2]{Gneiting2011}. Summarized in rough terms, one can construct realistic examples where the performance of  skilful statistical forecasts is ranked worse than an ignorant no-change forecast when evaluated by `some' scoring function, such as the absolute error or the squared error, for example. Therefore, point forecasts for a functional $\nu$ have to be evaluated by means of a scoring function, which is consistent for $\nu$.

\begin{defn}\label{def:elicit}
A functional $\nu$ is \emph{elicitable} relative to the class $\mathcal{P}$, if there exists a scoring function $s$ which is strictly consistent for $\nu$ relative to $\mathcal{P}$.
\end{defn}

Many interesting functionals are elicitable and a wealth of examples is given in \citet{Gneiting2011}. The most prominent example concerning risk management may be VaR, which is essentially a quantile and as such elicitable. The scoring functions that are consistent for $\alpha$-quantiles have been characterized by \citet{Thomson1979,Saerens2000}; see also \citet[Theorem 9]{Gneiting2011}. Subject to some regularity and integrability conditions, they are given by
\begin{equation}\label{eq:quantscore}
s(x,y) = (\mathbbm{1}\{x\ge y\} - \alpha)(g(x) - g(y)),
\end{equation}
where $g$ is an increasing function and $\mathbbm{1}$ denotes the indicator function.

However, not all functionals are elicitable, the most striking example in the present context being ES. The following necessary condition is due to \citet{Osband1985}; see also \citet{LambertPennockETAL2008}. As this theorem is central to the results presented in this paper, we provide a proof.

\begin{thm}[Osband]\label{thm:elicconv}
An elicitable functional $\nu$ has \emph{convex level sets} in the following sense: If $P_0 \in \mathcal{P}$ and $P_1 \in \mathcal{P}$, and $P^* = pP_0 + (1-p)P_1 \in \mathcal{P}$ for some $p \in (0,1)$, then $t \in \nu(P_0)$ and $t \in \nu(P_1)$ imply $t \in \nu(P^*)$.
\end{thm}
\begin{proof}
Let $s$ be a strictly consistent scoring function for $\nu$ and let $P_0$, $P_1$, $P^*$, $p$ and $t$ be as required in the Theorem. Then we obtain for any $x \in \mathbb{R}$ that
\begin{align*}
\E_{P^*} s(t,Y) & = p\E_{P_0} s(t,Y) + (1-p)\E_{P_1} s(t,Y) \\&\le p\E_{P_0} s(x,Y) + (1-p)\E_{P_1} s(x,Y) = \E_{P^*} s(x,Y),
\end{align*}
hence $t \in \nu(P^*)$.
\end{proof}

\section{Coherent risk measures}\label{sec:riskmeasures}

Let $(\Omega,\mathcal{F},\prob)$ be a standard probability space without atoms. A \emph{coherent risk measure} is a map $\rho:L^{\infty}(\Omega,\mathcal{F},\prob) \to \mathbb{R}$, which fulfils the following four properties. It is \emph{monotone}, so $X \ge Y$ implies $\rho(X) \le \rho(Y)$; it is \emph{subadditive}, that is $\rho(X+Y) \le \rho(X) + \rho(Y)$ for all $X,Y \in L^{\infty}(\Omega,\mathcal{F},\prob)$; it is \emph{positively homogeneous}, i.e.~for $\lambda \ge 0$, it holds that $\rho(\lambda Y) = \lambda \rho(Y)$. Finally, it is \emph{translation invariant} in the sense that for all $a \in \mathbb{R}$, we have $\rho(Y + a) = \rho(Y)-a$. 

Coherent risk measures were introduced by \citet{ArtznerDelbaenETAL1999}; see also \citet{Delbaen2002} and \citet[Chapter 4]{FollmerSchied2004}. All common coherent risk measures used in applications, share the property of \emph{law-invariance}. That is, if $X$ and $Y$ have the same distribution $P$ on $\mathbb{R}$, then
\[
\rho(P):= \rho(X) = \rho(Y).
\]
Law-invariant risk measures were characterized by \citet{Kusuoka2001}. His result was strengthened by \citet{JouiniSchachermETAL2006}. We summarize the result that is relevant for this paper in the following theorem; compare \citet[Theorem 2.1]{JouiniSchachermETAL2006}. Let $\mathfrak{P}$ denote the set of all probability measures on $[0,1]$ with the weak topology. For a cumulative distribution function $F:\mathbb{R} \to [0,1]$ we define its \emph{generalized inverse} or \emph{quantile function} by
\[
F^{-1}:[0,1] \to \mathbb{R}, \quad \alpha \mapsto \inf\{x \in \mathbb{R}\;|\; F(x) \ge \alpha\}.
\] 
\begin{thm}[Kusuoka]
Let $\rho$ be a law-invariant coherent risk measure. Then there exists a closed convex set $\mathcal{M} \subset \mathfrak{P}$ such that
\[
\rho(Y) = - \inf_{m \in \mathcal{M}} \int_{[0,1]} U_{\alpha}(Y) dm(\alpha),
\]
where 
\[
U_{\alpha}(Y) := \frac{1}{\alpha}\int_0^\alpha F_Y^{-1}(u) du, \quad \alpha \in (0,1],
\]
and $U_0(Y) = \essinf Y$.
\end{thm}

For $m \in \mathfrak{P}$, $Y \in L^{\infty}(\Omega,\mathcal{F},\prob)$ we define
\begin{equation}\label{eq:num}
\nu_m(Y) = \int_{[0,1]} U_{\alpha}(Y) m(d\alpha).
\end{equation}
Up to the sign, the $\nu_m$ are exactly the \emph{spectral risk measures} of \citet{Acerbi2002}.
The following alternative representation of $\nu_m$ will be useful in the following. It is a direct consequence of Fubini's theorem.
\begin{align}
\nu_m(Y) &= \int_{(0,1]} U_{\alpha}(Y)m(d\alpha) + m(\{0\})U_0(Y)\nonumber\\
&= \int_{(0,1]} \frac{1}{\alpha}\int_{(0,1]} \mathbbm{1}\{u \le \alpha\} F_Y^{-1}(u) du\ m(d\alpha) + m(\{0\})U_0(Y)\nonumber\\
&= \int_{(0,1]} \int_{[u,1]} \frac{1}{\alpha}  m(d\alpha) F_Y^{-1}(u) du + m(\{0\})U_0(Y)\nonumber\\
&= \E(g_m(U) F_Y^{-1}(U)) + m(\{0\})U_0(Y),\label{eq:altrep}
\end{align}
where $U$ is standard uniformly distributed, and $g_m:(0,1] \to [0,\infty)$ is given by
\[
g_m(u) = \int_{[u,1]} \frac{1}{\alpha}  m(d\alpha).
\]
Following \citet{PichlerShapiro2012} we call $g_m$ the \emph{spectral function} of $m \in \mathfrak{P}$. It is left-continuous, decreasing and $\int_0^1 g_m(u) du = m((0,1]) \le 1$. This implies in particular that the functional $\nu_m(Y)$ is finite for all $Y \in L^{\infty}(\Omega,\mathcal{F},\prob)$. We call the function $p \mapsto \int_p^1g_m(u)du$ the \emph{integrated spectral function} of $m$. As $\nu_m$ is law-invariant, we will also write $\nu_m(P):=\nu_m(Y)$ if $Y$ has distribution $P$.

\section{Coherence and elicitability}\label{sec:coherelicit}

\subsection{Coherent functionals with convex level sets}
If a functional $\nu$ is not elicitable with respect to a class of probability distributions $\mathcal{P}_0$, then it cannot be elicitable with respect to any larger class $\mathcal{P} \supset \mathcal{P}_0$. In particular, if the functional $\nu$ does not have convex level sets in the sense of Theorem \ref{thm:elicconv} for some class $\mathcal{P}_0$, it will fail to have convex level sets for any larger class $\mathcal{P}$ containing $\mathcal{P}_0$. A simple class $\mathcal{P}^*$ of probability distributions, that has proven useful to show the violations of the necessary condition for elicitability in Theorem \ref{thm:elicconv} is the class of two-point distributions, that is
\begin{equation}\label{eq:Pstar}
\mathcal{P}^* = \{p\delta_x + (1-p)\delta_y \;|\; x, y \in \mathbb{R}, p \in [0,1]\},
\end{equation}
where $\delta_x$ is the Dirac measure at the point $x \in \mathbb{R}$.

The following theorem summarizes the main results of this paper.

\begin{thm}\label{thm:1}
Let $\mathcal{P}^*$ be the class of two-point distributions on $\mathbb{R}$; see \eqref{eq:Pstar}. Let $\mathcal{M}\subseteq \mathfrak{P}$ be a closed set of probability measures on $[0,1]$, that does not contain $\delta_0$. If the functional
\[
\nu:\mathcal{P}^* \to \mathbb{R}, \quad P \mapsto \nu(P) = \inf_{m\in\mathcal{M}}\nu_m(P),
\]
with $\nu_m$ defined at \eqref{eq:num}, has convex level sets, then $m(\{0\})=0$ for all $m \in \mathcal{M}$, and
there exists a $C \in (0,1]$ such that all probability measures
\begin{equation}\label{eq:mp}
m_p = \frac{p(1-C)}{p(1-C) + C}\delta_p + \frac{C}{p(1-C) + C}\delta_1, \quad p \in (0,1),
\end{equation}
are contained in $\mathcal{M}$. Furthermore, for all measures $m \in \mathcal{M}$, there exists a $q \in (0,1)$ such that
\begin{equation}\label{eq:char}
\int_p^1 g_m(v) dv \ge \int_p^1 g_{m_q}(v) dv, \quad p \in [0,1],
\end{equation}
and 
\begin{equation}\label{eq:convorder}
\int_p^1 g_m(v) dv \ge \frac{C(1-p)}{C(1-p) + p}, \quad p \in [0,1].
\end{equation}
\end{thm}
The lower bound in \eqref{eq:convorder} and the integrated spectral functions of $m_p$ are illustrated in Figure \ref{fig:2}. We would like to give a short summary of the proof of Theorem \ref{thm:1}. The details are deferred to Section \ref{sec:proof}. 

For a two-point distribution $p\delta_x + (1-p)\delta_y$ it is possible to calculate that $\nu(p\delta_x + (1-p)\delta_y) = \alpha_p x + (1-\alpha_p)y$ for some $\alpha_p \in [0,1]$. It follows from the properties of $\nu$ as a risk measure that $\alpha_p > 0$. The assumption $\delta_0 \not\in\mathcal{M}$ is necessary to guarantee that for some $p \in (0,1)$ we have $\alpha_p < 1$. Exploiting the convexity of level sets as given by Theorem \ref{thm:elicconv}, first we show that $\alpha_p \in (0,1)$ for all $p \in (0,1)$, and then we derive an explicit formula for $\alpha_p$ in terms of $C$ and $p$. 

\begin{rem}
\citet{Weber2006} also studied law-invariant risk measures with convex level sets. His motivation came from considerations of dynamic consistency of risk measures, which he shows to be closely related to convexity of level sets. \citet[Theorem 3.1]{Weber2006} is more general than Theorem \ref{thm:1} in the sense that not only coherent risk measures are considered and that it provides a characterization instead of a necessary condition. However, his result requires regularity assumptions on the risk measure under consideration, which we do not impose in this work. See also the remarks in Section \ref{sec:spectral} and \ref{sec:expectiles}.
\end{rem}

\subsection{Spectral risk measures}\label{sec:spectral}

The first corollary to Theorem \ref{thm:1} shows that none of the coherent risk measures considered in \citet[Example 3]{PichlerShapiro2012} are elicitable, unless they reduce to minus the expected value.

\begin{cor}\label{cor:1}
Suppose the functional $\nu$ in Theorem \ref{thm:1} has convex level sets and there is a finite set $\mathcal{M}_0 \subset \mathfrak{P}$, $\delta_0 \not\in \mathcal{M}_0$ such that 
\[
\nu(P) = \inf_{m \in \mathcal{M}_0} \nu_m(P), \quad P \in \mathcal{P}^*,
\]
then $\mathcal{M}_0 =\{\delta_1\}$ and 
\[
\nu(P) = \E_P(Y).
\] 
\end{cor}
\begin{proof}
By Theorem \ref{thm:1} the closed set $\mathcal{M}_0$ is uncountable unless $C=1$. If $C=1$, then the lower bound in \eqref{eq:convorder} is $1-p$, which is the integrated spectral function of $\delta_1$. The claim follows directly from \citet[Lemma 2.2]{Dana2005}.
\end{proof}

Now, it follows easily that elicitable spectral risk measures are essentially the expected value.  

\begin{cor}\label{cor:2}
Spectral risk measures, other than minus the expected value, are not elicitable relative to any class of probability distributions that contains the two-point distributions.
\end{cor}
\begin{proof}
It is a direct consequence of Theorem \ref{thm:elicconv} and Corollary \ref{cor:1} that any spectral risk measure, which is not minus the essential infimum, is not elicitable unless it is minus the expected value. Therefore, it only remains to show that $\essinf(Y)$ is not an elicitable functional relative to the class of two-point distributions. Suppose the contrary, and let $S$ be a strictly consistent scoring function. If $Y = a$ almost surely for some $a \in \mathbb{R}$, then $\mathbb{E}S(x,Y) = S(x,a)$ and we obtain $S(a,a) < S(x,a)$ for all $x \in \mathbb{R}\setminus\{a\}$. If $Y$ has distribution $p\delta_a + (1-p)\delta_b$ with $a < b$ and $p \in (0,1]$, we obtain
\[
pS(x,a) + (1-p)S(x,b) > pS(a,a) + (1-p)S(a,b), \quad x \in  \mathbb{R}\setminus\{a\}.
\]
With $x = b$ and letting $p \to 0$, we obtain $S(b,b) \ge S(a,b)$, a contradiction.
\end{proof}

\begin{rem}
While the proof of Corollary \ref{cor:2} shows that $\essinf(Y)$ is not elicitable relative to the class of two-point distributions, it is easy to check that the interval $(-\infty,\essinf(Y)]$ is elicitable. Strictly consistent scoring functions are given at \eqref{eq:quantscore} with $\alpha = 0$ and any strictly increasing function $g$.
\end{rem}

\citet[Theorem 11]{Gneiting2011} shows that ES is not elicitable with respect to any class of probability measures that contains the measures with finite support, or the finite mixtures of absolutely continuous distributions with compact support. In both cases the proof is done by showing a violation of the necessary condition of convex level sets given in Theorem \ref{thm:elicconv}. We believe that it is possible to modify the proof of Theorem \ref{thm:1} using mixtures of absolutely continuous distributions with compact support instead of two-point distributions. However, the details remain to be worked out and are likely to be rather technical.

\begin{rem}
While the result that ES is not elicitable is due to \citet{Gneiting2011}, the non-convexity of its level sets already appears in \citet[Example 3.4]{Weber2006}.
\end{rem}

\subsection{Expectiles}\label{sec:expectiles}

Theorem \ref{thm:1} provides an upper and a lower bound on a potentially elicitable law-invariant coherent risk measures $\rho$ via the provided restrictions on the integrated spectral functions. This is illustrated in Figure \ref{fig:2}, and details are given below.

Let $\mathcal{M}\subset \mathfrak{P}$ be a closed set such that
\[
\rho(X) = - \inf_{m \in \mathcal{M}} \nu_m(X), \quad X \in L^{\infty}(\Omega,\mathcal{F},\mathbf{P}).
\]
By \citet[Lemma 2.2]{Dana2005} equation \eqref{eq:convorder} of Theorem \ref{thm:1} implies that there exists a $C \in (0,1]$ such that
\[
\rho(X) \le \mathfrak{u}_C(X) := -\int_0^1 \frac{C}{(v+C(1-v))^2}F_X^{-1}(v)dv.
\]
The map $\mathfrak{u}_C:L^{\infty}(\Omega,\mathcal{F},\mathbf{P}) \to \mathbb{R}$ is a spectral risk measure with spectral function $g(v) = C/(v+C(1-v))^2$ for $v \in (0,1]$. The associated measure $m \in \mathfrak{P}$ has density $2C(1-C)v/(v+C(1-v))^3$ on $(0,1)$ and a point mass $C$ at $v = 1$. The integrated spectral function $\int_p^1 g(v)dv$ is illustrated as a dashed line in Figure \ref{fig:2}. By Corollary \ref{cor:2}, the spectral risk measure $\mathfrak{u}_C$ is not elicitable unless $C = 1$. In this case, it reduces to minus the expected value. 

We define
\[
\mathcal{M}_C := \{m_p \in \mathfrak{P}\;|\; p \in (0,1)\},
\]
where $m_p$ is given at \eqref{eq:mp} and 
\begin{equation}\label{eq:lC}
\mathfrak{l}_C(X) := -\inf_{m \in \mathcal{M}_C} \nu_m(X), \quad X \in L^{\infty}(\Omega,\mathcal{F},\mathbf{P}). 
\end{equation}
By Theorem \ref{thm:1} we immediately obtain $\rho(X) \ge \mathfrak{l}_C(X)$. Invoking \citet[Lemma 2.2]{Dana2005} equation \eqref{eq:char} yields $\rho(X) \le \mathfrak{l}_C(X)$, hence $\rho(X) = \mathfrak{l}_C(X)$.
In the remainder of this section we characterize the law-invariant coherent risk measure $\mathfrak{l}_C$.

As introduced in \citet{NeweyPowell1987}, the \emph{$\tau$-expectile} $\mu_{\tau}(X)$, $\tau \in (0,1)$, of a random variable with finite mean is the unique solution $x = \mu_{\tau}(X)$ to the equation
\begin{equation}\label{eq:expectile}
\tau \int_x^{\infty} (y-x) dF_X(y) = (1-\tau)\int_{-\infty}^x (x-y) dF_X(y).
\end{equation}
\citet{BelliniKlarETAL2013} show that (up to the sign) expectiles are law-invariant coherent risk measures for $\tau \in (0,1/2]$. As mentioned in introduction, expectiles are elicitable. The scoring functions that are consistent for $\tau$-expectiles were recently characterized by \citet[Theorem 10]{Gneiting2011}. Subject to some regularity and integrability conditions, they are given by
\[
s(x,y) = |\mathbbm{1}\{x\ge y\} - \tau|(g(y) - g(x) - g'(x)(y-x)),
\]
where $g$ is a convex function with subgradient $g'$.
The prominent role of expectiles as the only elicitable law-invariant coherent risk measures is underlined by the following proposition.

\begin{prop}\label{prop:4.4}
The law-invariant coherent risk measure $\mathfrak{l}_C$ defined at \eqref{eq:lC} is minus the $\tau$-expectile for 
\[
\tau := \frac{C}{C+1} \;\in (0,\textstyle{\frac{1}{2}}].
\]
\end{prop}
\begin{proof}
Let $X$ a random variable with finite first moment, $F:=F_X$, and $\mu := \mu_{\tau}(X)$ its $\tau$-expectile with $\tau=C/(C+1)$. We define
\[
p^* := F(\mu).
\]
We will show that $\nu_{m_{p^*}}(X) = \mu$ and that $\nu_{m_p}(X)$ is minimal at $p = p^*$. If $F$ is continuous the latter claim can alternatively be shown by methods of calculus. We show the claims directly in order to avoid case distinctions. 

For $p \in (0,1)$, we obtain with $m_p$ defined at \eqref{eq:mp}
\begin{align*}
\nu_{m_p}(X) &= \frac{1}{p(1-C)+p}\int_0^p F^{-1}(v) dv + \frac{C}{p(1-C)+p}\int_p^1 F^{-1}(v) dv\\
&= \frac{1}{p(1-C)+p}\int_{-\infty}^{F^{-1}(p)} y dF(y) + \frac{C}{p(1-C)+p}\int_{F^{-1}(p)}^{\infty} y dF(y)\\&\quad + \frac{1-C}{p(1-C)+p}F^{-1}(p)(p - F(F^{-1}(p))),
\end{align*}
where we used \citet[Proposition 3.2]{AcerbiTasche2002} in the second step. 

\citet[Theorem 1]{NeweyPowell1987} show that $p^* \in (0,1)$. Using \eqref{eq:expectile} we obtain
\begin{align}
\int_{-\infty}^{\mu}& y dF(y) + C\int_{\mu}^{\infty}ydF(y)\nonumber \\&= \int_{-\infty}^{\mu} y dF(y) + C\int_{\mu}^\infty \mu dF(y) + \int_{-\infty}^{\mu} \mu dF(y) - \int_{-\infty}^{\mu} y dF(y)\nonumber\\
&= C\mu(1 - F(\mu)) + \mu F(\mu) = \mu(p^*(1-C) + C). \label{eq:integrals}
\end{align}
Let $\varepsilon \ge 0$ such that $p:= p^* + \varepsilon \in (0,1)$. Then, using equation \eqref{eq:integrals} and partial integration, we obtain
\begin{align*}
\nu_{m_p}(X) &= \frac{1}{p(1-C)+C}\Big(\int_{-\infty}^\mu y dF(y) + C\int_{\mu}^{\infty}y dF(y)\\&\qquad+ \int_{\mu}^{F^{-1}(p)}ydF(y) - C\int_{\mu}^{F^{-1}(p)}y dF(y) \\&\qquad+ (1-C)F^{-1}(p)(p - F(F^{-1}(p)))\Big)\\
&= \frac{\mu(p^*(1-C)+C)}{p(1-C) + C}\\&\quad + \frac{1-C}{p(1-C)+C}\Big( F^{-1}(p)F(F^{-1}(p))-\mu F(\mu) \\&\qquad- \int_{\mu}^{F^{-1}(p)}F(t)dt + F^{-1}(p)p - F^{-1}(p)F(F^{-1}(p))\Big)\\
&= \mu - \frac{\mu\varepsilon(1-C)}{p(1-C) + C}\\
&\quad+\frac{1-C}{p(1-C)+C}\Big(F^{-1}(p)p - \mu p + \mu \varepsilon - \int_{\mu}^{F^{-1}(p)}F(t)dt\Big)\\
&= \mu + \frac{1-C}{p(1-C)+C}\Big(p(F^{-1}(p) - \mu ) - \int_{\mu}^{F^{-1}(p)}F(t)dt\Big).
\end{align*}
The last term in the above equation is always non-negative. It vanishes for $p = p^*$, hence
$\nu_{m_{p^*}}(X)\le \nu_{m_p}(X)$ for all $p \ge p^*$. The argument for $p \le p^*$ is completely analogous.
\end{proof}

\begin{rem}
Expectiles as coherent risk measures also appear implicitly in \citet[Corollary 3.2]{Weber2006}. He shows that the shortfall risk measure with loss function $\ell(x) = \alpha x^+ - \beta x^-$ for $\alpha \ge \beta > 0$ is coherent. Such a shortfall risk measure is equal to the minus the $\tau$-expectile with $\tau=\beta/(\alpha+\beta)$. However, \citet{Weber2006} did not draw the connection to the expectiles \citep{NeweyPowell1987} in the statistical literature. Under the additional regularity assumptions (3.1) and (1) of \citet[Theorem 3.1]{Weber2006}, \citet[Corollary 3.1]{Weber2006} characterizes all coherent risk measures with convex level sets as minus $\tau$-expectiles with $\tau \in (0,1/2]$. Theorem \ref{thm:1} shows that these conditions are not necessary for the characterization in the coherent case.
\end{rem}

\begin{figure}
\begin{center}
\includegraphics{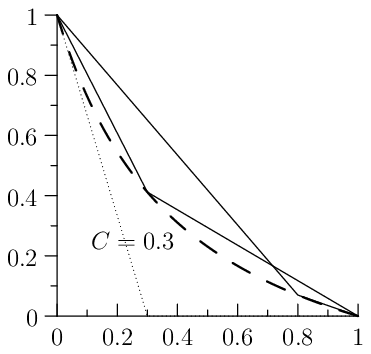}
\includegraphics{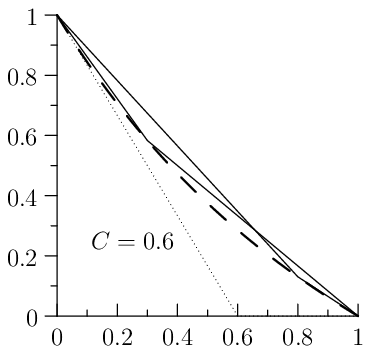}
\end{center}
\caption{Integrated spectral functions. In both panels, the dashed lines are the integrated spectral functions of $-\mathfrak{u}_C$, and the solid lines are those of $\nu_{m_{0.3}}$ and $\nu_{m_{0.8}}$ as examples. For comparison, the dotted line is the integrated spectral function of $-ES_C$. \label{fig:2}}
\end{figure}

\subsection{Proof of Theorem \ref{thm:1}}\label{sec:proof}
For each $m \in \mathcal{M}$, $0 < p_1 \le p_2 \le 1$, we define
\[
P_m(p_1,p_2) =  \int_{p_1}^{p_2}g_m(v)dv, \quad P_m(0,p_2) = \int_0^{p_2}g_m(v)dv + m(\{0\})
\]
Using Fubini we obtain
\begin{align}
P_m(p_1,p_2) &= \int_{p_1}^{p_2} \int_{[v,1]}\frac{1}{\alpha}m(d\alpha) dv\nonumber\\
&= \int_{[p_1,1]} \frac{1}{\alpha} \int_{p_1}^{\alpha\wedge p_2} dv m(d\alpha)= \int_{[p_1,1]} \frac{\alpha\wedge p_2-p_1}{\alpha}  m(d\alpha) \label{ex:4.4}
\end{align}
For all $q \in (0,1]$, the set 
\[
\mathcal{C}_q = \{(P_m(0,q),P_m(q,1))\in \mathbb{R}^2\;|\; m \in \mathcal{M}\}
\]
is a subset of the unit simplex in $\mathbb{R}^2$ because $P_m(0,q) + P_m(q,1) =  1$. The set $\mathcal{C}_q$ is also closed, which can be seen using Helly's theorem, the fact that $\mathcal{M}$ is closed, and the representation of of $P_m(q,1)$ at \eqref{ex:4.4}.
Let $\alpha_q:= \sup\{\alpha \in [0,1]\;|\; (\alpha,1-\alpha)\in \mathcal{C}_q\}$ be the lower boundary point of $\mathcal{C}_q$; see Figure \ref{fig:1} for an illustration. As $\mathcal{C}_q$ is closed, the supremum is attained and there is an $m_q \in \mathcal{M}$ such that $P_{m_q}(0,q) = \alpha_q$. Note that $\alpha_q > 0$. Suppose the contrary, then 
\[
\int_0^{q}g_{m_q}(v)dv + m_q(\{0\}) = 0,
\]
which implies $m_q(\{0\})=0$ and $g_{m_q}(v) = 0$ for $v \in (0,q]$, hence $g_{m_q}\equiv 0$ because $g_{m_q}$ is decreasing and non-negative. This is a contradiction because 
\[
\int_0^1 g_m(v) dv  + m(\{0\}) = m([0,1]) = 1.
\] 
The function $(0,1] \to (0,1], q \mapsto \alpha_q$ is increasing.
Define
\begin{equation}\label{eq:qstar}
q^* := \inf \{q \in (0,1]\;|\; \alpha_q = 1\}.
\end{equation}
If $q^* = 0$, then $\alpha_q = 1$ for all $q \in (0,1]$, hence $1-\alpha_q = \int_q^1 g_{m_q}(v)dv = 0$. This implies $m_q((q,1]) = 0$ for all $q \in (0,1]$, and hence $m_q$ converges weakly to $\delta_0$ as $q \to 0$. As $\mathcal{M}$ is closed this is a contradiction to the assumption $\delta_0 \not\in\mathcal{M}$. Therefore, $q^* > 0$. We will conclude later that, actually, $q^*=1$.

\begin{figure}
\centering
\includegraphics{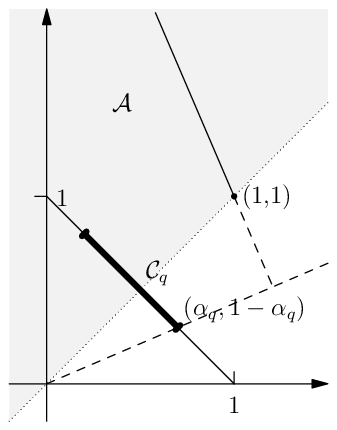}
\caption{Illustration of the construction in the proof of Theorem \ref{thm:1}.\label{fig:1}}
\end{figure}

Let $\mathcal{A} = \{x = (x_1,x_2) \in \mathbb{R}^2\;|\; x_1 \le x_2\}$. For $(x_1,x_2) \in \mathcal{A}$, $p \in [0,1]$ the distribution function $F$ of $p\delta_{x_1} + (1-p)\delta_{x_2}$ has generalized inverse
%\[
%F(t) = p\mathbbm{1}\{t \ge x_1\} + \mathbbm{1}\{t \ge x_2\}, \quad t \in \mathbb{R}
%\]
$F^{-1}(v) = x_1\mathbbm{1}\{v \le p\} + x_2\mathbbm{1}\{v > p\}$, $v \in (0,1]$.
This yields using \eqref{eq:altrep}
\begin{align*}
\nu_m(p\delta_{x_1} + (1-p)\delta_{x_2}) &= x_1\int_0^p g_m(v) dv + x_2 \int_p^1 g_m(v) dv + m(\{0\})x_1\\
&= P_m(0,p) x_1 + P_m(p,1)x_2,
\end{align*}
hence
\begin{align*}
\nu(p\delta_{x_1}+(1-p)\delta_{x_2}) &= \inf_{m \in \mathcal{M}} \nu_m(p\delta_{x_1}+(1-p)\delta_{x_2})\\ &= \inf_{(\alpha,1-\alpha) \in \mathcal{C}_p} \big(\alpha x_1 + (1-\alpha) x_2\big)\\ &= \alpha_p x_1 + (1-\alpha_p) x_2 = \nu_{m_p}(p\delta_{x_1}+(1-p)\delta_{x_2}).
\end{align*}

Let $x_1 = x_2 = 1$, $q \in (0,q^*)$. Then $0<\alpha_q < 1$. All $y = (y_1,y_2) \in \mathcal{A}$ with $\nu(\delta_1) = 1 = \nu(q\delta_{y_1}+(1-q)\delta_{y_2})$ are given by
\[
y_1 = 1 - c, \quad y_2 = M y_1 + a
\]
for $M = -\alpha_q/(1-\alpha_q) \in (-\infty,0)$, $a = 1 - M$ and $c \ge 0$; cf.~Figure \ref{fig:1}. Note that
$y_1 \le x_1 \le x_2 \le y_2$.
Convexity of the level sets of $\nu$ implies that for all $m \in \mathcal{M}$, all $v \in (0,1]$ and all $c > 0$, we have
\begin{align}
1 &\le \nu_{m}\big(v(q\delta_{y_1}+ (1-q)\delta_{y_2}) + (1-v)(q\delta_{x_1} + (1-q)\delta_{x_2}) \big)\nonumber \\
&= P_m(0,qv) y_1 + P_m(qv,q) x_1 + P_m(q,1 - v + qv) x_2 + P_m(1 - v + qv,1) y_2\nonumber\\
&= \nu_m(\delta_1) + P_m(0,qv) (y_1-x_1) + P_m(1 - v + qv,1) (y_2-x_2)\nonumber\\
&= 1 - cP_m(0,qv) - cMP_m(1 - v + qv,1),\nonumber
\end{align}
hence, 
\begin{equation}\label{eq:34}
P_m(0,qv) \le \frac{\alpha_q}{1-\alpha_q}P_m(1 - v + qv,1).
\end{equation}
We have $\lim_{v \downarrow 0}P_m(0,qv) = m(\{0\})$, and $\lim_{v \downarrow 0} P_m(1 - v + qv,1) = P_m(1,1) = 0$,
hence it follows that $m(\{0\})=0$ for all $m \in \mathcal{M}$. Equation \eqref{eq:34} also implies that $g_m(v) > 0$ for all $v \in (0,1)$, $m \in \mathcal{M}$. Suppose the contrary. Then there is a $w^* \in(0,1)$ such that $g_m(v) = 0$ for all $v \in (w^*,1]$, which implies $P_m(1-v^* + qv^*,1)=0$ for some $v^* \in (0,1)$. Now \eqref{eq:34} yields $g_m \equiv 0$, which is a contradiction because $m(\{0\})=0$. 
Going back to the definition of $q^*$ at \eqref{eq:qstar} we obtain in particular that $1-\alpha_p = \int_p^1 g_{m_p}(v) dv > 0$ for all $p \in (0,1)$. Therefore, $q^* = 1$.

For $m = m_q$ we obtain
\[
P_{m_q}(0,qv) \le \frac{P_{m_q}(0,q)}{P_{m_q}(q,1)}P_{m_q}(q + (1-q)(1-v),1),
\]
hence 
\begin{multline*}
P_{m_q}(0,qv)\big(P_{m_q}(q,1 - v + qv)+P_{m_q}(1 - v + qv,1)\big) \\\le \big(P_{m_q}(0,qv)+P_{m_q}(qv,q)\big)P_{m_q}(1 - v + qv,1)
\end{multline*}
which yields
\[
P_{m_q}(0,qv)P_{m_q}(q,1 - v + qv) \le P_{m_q}(qv,q)P_{m_q}(1 - v + qv,1).
\]
Monotonicity of $g_m$ yields
\begin{align*}
qv g_{m_q}(qv) (1-q)(1-v)&g_{m_q}(1 - v + qv) \le P_{m_q}(0,qv)P_{m_q}(q,1 - v + qv)\\
&\le P_{m_q}(qv,q)P_{m_q}(q + (1-q)(1-v),1) \\&\le q(1-v)g_{m_q}(qv+) v(1-q) g_{m_q}((1 - v + qv)+),
\end{align*}
hence we obtain
\begin{equation}\label{eq:Pmq}
P_{m_q}(0,qv)P_{m_q}(q,1 - v + qv) = P_{m_q}(qv,q)P_{m_q}(1 - v + qv,1).
\end{equation}
Equation \eqref{eq:Pmq} implies that
\begin{equation}\label{eq:cq1}
c_q := \frac{\alpha_q}{1-\alpha_q} = \frac{P_{m_q}(0,q)}{P_{m_q}(q,1)} = \frac{P_{m_q}(0,qv)}{P_{m_q}(1 - v +qv,1)} 
= \frac{P_{m_q}(qv,q)}{P_{m_q}(q,1 - v +qv)},
\end{equation}
therefore
\begin{align*}
\frac{q}{1-q}\frac{g_{m_q}(qv)}{g_{m_q}((1-v+qv)+)}&\le \frac{P_{m_q}(0,qv)}{P_{m_q}(1 - v +qv,1)} = c_q = \\
&= \frac{P_{m_q}(qv,q)}{P_{m_q}(q,1 - v +qv)} \le \frac{q}{1-q}\frac{g_{m_q}(qv+)}{g_{m_q}(1-v+qv)},
\end{align*}
and hence
\begin{equation}\label{eq:cq}
c_q(1-q) g_{m_q}((1-v(1-q))+) = q g_{m_q}(qv)
\end{equation}
for all $v \in (0,1]$. The left-hand side is increasing in $v$, whereas the right-hand side is decreasing in $v$. Both sides are left-continuous. This implies that
\[
g_{m_q}(w) = \begin{cases}c_1, & w \in (0,q],\\
c_2, & w \in (q,1], \end{cases}
\]
for two constants $c_1 \ge c_2$. We have
\[
1 = m_q((0,1]) = \int_0^1 g_m(v) dv = qc_1 + (1-q)c_2
\]
and by \eqref{eq:cq}
\[
c_q(1-q)c_2 = q c_1,
\]
hence
\[
1 = c_q(1-q)c_2 + (1-q)c_2 = c_2(1-q)(1+c_q)
\]
which yields
\begin{equation}\label{eq:gmq}
g_{m_q}(w) = \begin{cases}\alpha_q/q, & w \in (0,q],\\
(1-\alpha_q)/(1-q), & w \in (q,1]. \end{cases}
\end{equation}
The inequality $c_1 \ge c_2$ is equivalent to $\alpha_q \ge q$.
By the definition of $g_{m_q}$ and \eqref{eq:gmq} we obtain that
\[
m_q = d_{q,1} \delta_q + d_{q,2} \delta_1,
\]
where $d_{q,1}$, $d_{q,2}$ fulfil $d_{q,1} + d_{q,2} = 1$,
\[
d_{q,2} = g_{m_q}(1) = \frac{1-\alpha_q}{1 - q}.
\]
Therefore
\[
d_{q,1} = 1- d_{q,2} = \frac{\alpha_q-q}{1 - q}.
\]

For any $m \in \mathcal{M}$ equation \eqref{eq:34} implies
\begin{equation}\label{eq:14}
\frac{q}{1-q}\frac{g_m(qv)}{g_m((1-v + qv)+)} \le \frac{P_m(0,qv)}{P_m(1-v+qv,1)} \le \frac{\alpha_q}{1-\alpha_q}.
\end{equation}
If $p < q \in (0,1)$, let $v = p/q$. Then the above inequality implies
\[
\frac{q}{1-q} \frac{\alpha_p(1-p)}{(1-\alpha_p)p} \le \frac{\alpha_q}{1-\alpha_q}.
\]
On the other hand, we also obtain with $v < (1-p)/(1-q)$, that
\[
\frac{p}{1-p} \frac{\alpha_q(1-q)}{(1-\alpha_q)q} \le \frac{\alpha_p}{1-\alpha_p}.
\]
Hence for $p,q \in (0,q^*)$ we obtain that
\[
\frac{(1-\alpha_q)q}{\alpha_q(1-q)} = \frac{(1-\alpha_p)p}{\alpha_p(1-p)} =: C \in (0,1].
\]
If $C = 0$, this implies $\alpha_p = 1$, which is a contradiction.
This means that we can express $\alpha_p$ in terms of $C$ and $p \in (0,1)$ as
\[
\alpha_p = \frac{p}{C(1-p) + p}.
\]

Let $m \in \mathcal{M}$. To show equation \eqref{eq:char}, observe that $m(\{0\})=0$ implies that $\int_0^1 g_m(v)dv = 1$. Therefore, it suffices to show that there exists $q \in (0,1)$ with $g_m(0+)\le g_{m_q}(0+)$ and $g_m(1) \ge g_{m_q}(1)$ due to the convexity of the integrated spectral function $p \mapsto \int_p^1 g_m(v)dv$ and the piecewise linearity of $p \mapsto \int_p^1 g_{m_q}(v)dv$; cf.~Figure \ref{fig:2}.

Equation \eqref{eq:14} implies 
\[
g_m(pv) \le \frac{1}{C}\, g_m(1 - v + pv)
\]
for $p \in (0,1)$, $v \in (0,1]$. Taking the limit as $v \downarrow 0$ yields
\[
1 \le g_m(0+) \le \frac{1}{C}\, g_m(1) \le \frac{1}{C}.
\]
Suppose first that $C \in (0,1)$. In this case we have that $g_m(0+) < 1/C$. If we suppose on the contrary that $g_m(0+) = 1/C$ it follows that $g_m(1) = 1$, hence $m = \delta_1$, which in turn implies $g_m(0+)= 1$, a contradiction. If $g_m(0+) = 1$, then \eqref{eq:char} holds for all $q \in (0,1)$. If $g_m(0+) > 1$, then there exists a $q \in (0,1)$ such that $g_{m_q}(0+) = 1/(C(1-q) + q) = g_m(0+)$. Furthermore,
it follows that
\[
C g_m(0+) = \frac{C}{C(1-q) + q} = g_{m_q}(1) \le g_m(1).
\]
The case $C=1$ is easy.

The last claim of the theorem follows because
\[
1- \alpha_p = \frac{C(1-p)}{C(1-p)+p} = \inf_{m \in \mathcal{M}} \int_p^1 g_m(v) dv.
\]
\hfill \qed

\section{Discussion}\label{sec:disc}

In this paper we have shown that spectral risk measures are not elicitable, so it is unclear if and how it is possible to rank different point forecasts for such measures in a decision theoretically sound manner. In other words, objective comparison of competing estimation procedures for spectral risk measures is difficult, if not impossible.  This does not imply that it is impossible to perform backtests for a specific estimation procedure under fixed model assumptions. However, if one needs to decide between two estimation methods, the values of the test statistic used for backtesting should not be used for a quality ranking of the methods. Such a ranking should only be done through consistent scoring functions, which do not exist for spectral risk measures. 

A possible solution to the problem could be working with probabilistic forecasts for $Y$ in the following way. Suppose that an empirical loss distribution $\hat{F}_Y$ for $F_Y$ has been estimated.  Usually, using this estimated distribution, the forecast $\rho(\hat{F}_Y)$ for the risk measure $\rho(Y) = \rho(F_Y)$ of $Y$ is calculated, and all further verification and backtesting of the model and estimation procedure are solely based on $\rho(\hat{F}_Y)$. Alternatively, one could directly assess the probabilistic forecast $\hat{F}_Y$; see for example \citet{DieboldGuntherETAL1998,Berkowitz2001}. Competing forecasts could be compared using proper scoring rules; see for example \citet{GneitingRaftery2007}. \citet{GneitingRanjan2011a} propose a method to compare density forecasts with emphasis on different regions of interest, such as the center or the tails of the distributions. If the aim is to accurately predict a functional focussed on the tails, such as ES, this approach seems promising. 

\citet{McNeilFrey2000} propose a statistic, called \emph{exceedance residuals}, for backtesting ES, which can be interpreted as a score, and is used as such in \citet[Section 5]{ChunShapiroETAL2012}; see also \citet[Section 4.4.3]{McNeilFreyETAL2005}. This score is not a scoring function in the sense of this paper, as it depends on two functionals of the future outcome $Y$, namely on $\VaR_{\alpha}(Y)$ and $\ES_{\alpha}(Y)$. Potentially, it is necessary and useful to extend the concept of elicitability in order to put their construction into a decision theoretic framework. A useful notion may be \emph{$k$-elicitability} as introduced in \citet[Definition 11]{LambertPennockETAL2008}.  It is an interesting open question whether the bivariate functional of $\VaR_{\alpha}$ and $\ES_{\alpha}$ is 2-elicitable. 

Finally, \citet{ContDeguestETAL2010} propose to rethink the necessity of the subadditivity axiom as it is in conflict with robustness of risk measurement procedures for spectral risk measures; cf.~Section \ref{sec:intro}. Supporting their suggestion, we believe that elicitability is also a crucial requirement that should be taken into consideration when choosing a risk measure, if probabilistic forecast evaluation is not an option. Given their appealing statistical properties, the potential of expectiles as risk measures should be further investigated.

\bibliographystyle{plainnat}
\bibliography{biblio}

\end{document}